\documentclass[11pt, a4paper]{article}

\title{Weisfeiler-Leman Invariant Promise Valued CSPs\thanks{This work
appeared in \emph{Proceedings of the 28th International Conference on Principles and Practice of Constraint Programming (CP 2022)}. Libor Barto has received funding from the European Research Council (ERC) under the European Unions Horizon 2020 research and innovation programme (Grant Agreement No.~771005, CoCoSym). Silvia Butti was supported by a MICCIN grant PID2019-109137GB-C22 and by a fellowship from ``la Caixa'' Foundation (ID 100010434). The fellowship code is LCF/BQ/DI18/11660056. This project has received funding from the European Union’s Horizon 2020 research and innovation programme under the Marie Skłodowska-Curie grant agreement No. 713673.}}

\author{Libor Barto\\
Department of Algebra, Faculty of Mathematics and Physics\\
Charles University, Czechia\\
\texttt{libor.barto@gmail.com}
\and
Silvia Butti\\
Department of Information and Communication Technologies\\
Universitat Pompeu Fabra, Spain\\
\texttt{silvia.butti@upf.edu}
}

\date{}

\usepackage[margin=1in]{geometry}
\usepackage{amsmath}
\usepackage{amsthm} 
\usepackage{amssymb}
\usepackage{hyperref}
\hypersetup{colorlinks=true,citecolor=black,linkcolor=black,urlcolor=black}
\usepackage{amsfonts}
\usepackage{mathtools}
\usepackage{thmtools}
\allowdisplaybreaks
\usepackage{empheq}
\usepackage{enumitem}

\newtheorem{theorem}{Theorem}
\newtheorem{proposition}[theorem]{Proposition}

\theoremstyle{definition}
\newtheorem{example}{Example}

\newcommand{\qq}{\mathbb{Q}_{\geq0}}
\newcommand{\qinfty}{\mathbb{Q}_{\infty}}
\newcommand{\mb}[1]{\mathbf{#1}}
\newcommand{\graph}[1]{\mathrm{G}_{#1}}
\newcommand{\csp}{\textnormal{CSP}}

\newcommand{\pcsp}{\textnormal{PCSP}}
\newcommand{\pvcsp}{\textnormal{PVCSP}}
\newcommand{\lbl}{\ell}
\newcommand{\thresh}{\tau}

\newcommand{\bb}{\textbf{b}}
\newcommand{\bp}{\textbf{p}}
\newcommand{\ba}{\textbf{a}}
\newcommand{\bv}{\textbf{v}}
\newcommand{\bu}{\textbf{u}}

\newcommand{\bfa}{\mathbf{A}}
\newcommand{\bfi}{\mathbf{I}}
\newcommand{\bfb}{\mathbf{B}}
\newcommand{\bfc}{\mathbf{C}}

\newcommand{\bfy}{\mathbf{Y}}
\newcommand{\bfj}{\mathbf{J}}
\newcommand{\ab}{(\bfa,\bfb)}
\newcommand{\ia}{(\bfi,\bfa)}
\newcommand{\ib}{(\bfi,\bfb)}

\newcommand{\yone}{\mathbf{Y}_{1}}
\newcommand{\ytwo}{\mathbf{Y}_{2}}

\newcommand{\sa}{\textnormal{SA}}
\newcommand{\saia}{\sa^{1}\ia}
\newcommand{\blp}{\textnormal{BLP}}

\newcommand{\RA}{R^\bfa}
\newcommand{\RB}{R^\bfb}
\newcommand{\RI}{R^\bfi}
\newcommand{\val}{\operatorname{Val}}
\newcommand{\opt}{\operatorname{Opt}}
\newcommand{\ar}{\operatorname{ar}}
\newcommand{\lpam}{\operatorname{LP}^{m}(\bfa)}
\newcommand{\eqone}{\equiv_{1}}
\newcommand{\quot}{/\negthickspace\equiv_1}
\newcommand{\ith}{^{\textnormal{th}}}
\newcommand{\mult}[1]{\{\!\{#1\}\!\}}
\newcommand{\yes}{\mathrm{Yes}}
\newcommand{\no}{\mathrm{No}}

\begin{document}

\maketitle

\begin{abstract}
In a recent line of work, Butti and Dalmau have shown that a  fixed-template Constraint Satisfaction Problem  is solvable by a certain natural linear programming relaxation (equivalent to the basic linear programming relaxation) if and only if it is   solvable on a certain distributed network, and this happens if and only if its set of $\yes$ instances is closed under Weisfeiler-Leman equivalence. We generalize this result to the much broader framework of fixed-template Promise Valued Constraint Satisfaction Problems. 
Moreover, we show that two commonly used linear programming relaxations are no longer equivalent in this broader framework. 
\end{abstract}

\section{Introduction}

The Constraint Satisfaction Problem (CSP) is the problem of deciding whether there is an assignment of values from some domain $A$ to a given set of variables, subject to constraints on the combinations of values which can be assigned simultaneously to certain specified subsets of variables; the allowed combinations of values are specified by relations on $A$. 

Many important computational problems, including various versions of logical satisfiability, graph coloring, and systems of equations, can be obtained by fixing a finite domain and restricting the set of allowed relations~\cite{feder1998computational,krokhin2017constraint}. The restrictions can be specified by fixing a relational structure $\bfa$, called a template. The CSP over $\bfa$ is then the CSP restricted to instances that use only relations in $\bfa$. For example, if $\bfa$ consists of a single binary relation $R^{\bfa} \subseteq A^2$, an instance of the CSP over $\bfa$ is, e.g.,
\begin{equation} \label{eq:CSPinst}
  R(x_1,x_2), \ R(x_3,x_1), \ R(x_2, x_4), \ R(x_3,x_3).
\end{equation}
The goal is to decide whether there exists an assignment $h:\{x_1, x_2, \dots\} \to A$ that satisfies all the constraints, that is, $(h(x_1),h(x_2)) \in R^{\bfa}$, 
$(h(x_3),h(x_1)) \in R^{\bfa}$, etc. (see Section~\ref{sec:prelim} for formal definitions). For instance, if $R^{\bfa}$ is the disequality relation $\neq$ on $A$, then the CSP over $\bfa$ is essentially the Graph $|A|$-Coloring Problem. 

This paper deals with CSPs over fixed templates with finite domains. In particular, the phrase ``a CSP'' in the following discussion means the CSP over some template.  

The (finite-domain, fixed-template) CSP has been a very active research area in the last 20 years, fueled by the tight connection between the complexity of a CSP and the polymorphisms of its template -- these are multivariate functions on the domain that preserve all relations in the template (see~\cite{Barto2017polymorphisms}). The highlight in the area is the dichotomy theorem \cite{Bulatov2017,Zhu20}: every CSP is either solvable in polynomial time or NP-complete (assuming P is not NP); moreover, the polynomial cases are characterized by means of polymorphisms. Other major results include characterizations of applicability of fundamental algorithms, e.g., certain convex relaxations~(see \cite{KolmogorovTZ15,ThapperZ17}). 

A natural linear programming relaxation, which is central in this paper, can be obtained by formulating a CSP instance as a feasibility problem for a zero-one integer program  and then relaxing the requirement that each variable $p$ is in $\{0,1\}$ to $p \in [0,1]$. In fact, there are two widely used relaxations of this form, the Basic Linear Programming ($\blp$) relaxation (see \cite{KZ_vcsps}) and a slightly stronger relaxation, which we denote by $\sa^1$ to highlight its connection to the Sherali-Adams hierarchy~\cite{Sherali1990} for CSPs (see~\cite{ButtiD21fractional}). The difference between the two relaxations is only in how they address repeated variables in a constraint. It turns out that both relaxations (correctly) decide the same CSPs~\cite{ButtiD21fractional} in the sense that, for any template $\bfa$, all instances of the CSP over $\bfa$ are decided by the $\sa^1$ relaxation if and only if they are decided by $\blp$.\footnote{
We remark that in the literature the difference between the two relaxations is sometimes neglected, which occasionally leads to unjustified or slightly incorrect claims.
}
Moreover, this happens if and only if the template admits symmetric polymorphisms of all arities~\cite{Kun2012} (see also~\cite{Barto2021algebraicPCSP}).

The class of CSPs decided by $\blp$ ($\sa^1$) has reappeared recently in~\cite{ButtiD2021distributed}, where it was shown that it coincides with the class of CSPs which can be solved on a distributed network. The distributed set-up here is based on the DCSP framework of Yooko et al.~\cite{yokoo1992distributed}; informally,  each constraint and each variable is controlled by an agent; the communication is only between a constraint and a variable that participates in it; and the agents are anonymous, they communicate in synchronous rounds, and they all run the same deterministic algorithm.

The papers~\cite{ButtiD2021distributed,ButtiD21fractional} contribute another interesting characterization, by means of an equivalence akin to the 1-dimensional Weisfeiler-Leman graph isomorphism test~\cite{leman1968reduction}. For two CSP instances $\bfi$, $\bfj$ we write $\bfi \eqone \bfj$ if, very roughly, they cannot be distinguished by considering their local structure around variables (number and type of constraints they participate in, number and type of constraints their adjacent variables participate in, and so on). Now the equivalent conditions discussed above are also equivalent to the CSP being invariant under $\eqone$. Altogether, we have the following theorem, which witnesses the significance of this class of CSPs.

\begin{theorem}[\cite{ButtiD2021distributed,ButtiD21fractional,Kun2012}] \label{thm:old}
The following are equivalent for the $\csp$ over a finite structure $\bfa$.
\begin{enumerate}[label=\textnormal{(\roman*)}]
    \item There exists a distributed algorithm that solves $\csp(\bfa)$. Moreover, in such a case, there is a polynomial-time distributed algorithm that solves $\csp(\bfa)$.  \label{item:1thmold}
    \item If two instances of $\csp(\bfa)$ are $\eqone$-equivalent, then they are either both $\yes$ instances or both $\no$ instances. \label{item:2thmold}
    \item $\sa^1$ decides $\csp(\bfa)$. \label{item:3thmold}
    \item $\blp$ decides $\csp(\bfa)$. \label{item:4thmold}
    \item $\bfa$ has symmetric polymorphisms of every arity. \label{item:5thmold}
\end{enumerate}
\end{theorem}

\noindent
Our main result generalizes Theorem~\ref{thm:old} to a much broader setting, which we introduce next. 

\subsection{Promise Valued CSP}
The framework of Valued CSP (VCSP) generalizes CSP as follows. Instead of relations we consider valued relations (also known as cost functions) -- mappings that assign to tuples rational or positive infinite costs. Returning to the example above, $R^{\bfa}$ is now a mapping from $A^2$ to $\mathbb{Q} \cup \{\infty\}$ instead of a subset of $A^2$. The objective of the search version of the VCSP over $\bfa$ is to minimize a sum, e.g.,
\begin{equation} \label{eq:VCSPinst}
  R(x_1,x_2) + \ R(x_3,x_1) + \ R(x_2, x_4) + \ R(x_3,x_3),
\end{equation}
that is, to find an assignment $h$ such that $R^{\bfa}(h(x_1),h(x_2)) + R^{\bfa}(h(x_3),h(x_1)) + \dots$ is minimal. In the decision version, which we consider in this paper, the instance is such a sum together with a rational number $\thresh$ and we aim to decide whether the minimum is at most $\thresh$.

Notice that (the decision version of) VCSP indeed generalizes CSP since relations can be modelled by $\{0,\infty\}$-valued relations. On the other hand, MaxCSP -- where the aim is to maximize the number of satisfied constraints given a CSP instance -- is exactly the VCSP over $\{0,1\}$-valued relational structures. The VCSP framework also includes many problems of a mixed optimization and combinatorial nature, such as the Vertex Cover Problem~(see \cite{KZ_vcsps}). The VCSP area is also well developed; for instance, the approach via an appropriate generalization of polymorphisms still works~(see~\cite{KZ_vcsps}), a dichotomy theorem is available~\cite{Kolmogorov2017}, and the equivalence of \ref{item:4thmold} and \ref{item:5thmold} in Theorem~\ref{thm:old} can be lifted as well~\cite{KolmogorovTZ15}.

The more recent framework of Promise CSP (PCSP) generalizes CSP in a different direction. Here the relations are ``crisp'' but the template is a pair of structures $(\bfa,\bfb)$ of the same signature. Intuitively, $R^{\bfa}$ is a ``strict'' form of $R$ and $R^{\bfb}$ is its ``relaxed'' form. The PCSP over $(\bfa,\bfb)$ is the problem of distinguishing instances solvable in $\bfa$ from those which are not solvable in $\bfb$. Note that the problem only makes sense if every instance solvable in $\bfa$ is also solvable in $\bfb$ (this is equivalent to $\bfa$ being homomorphic to $\bfb$). A well-known family of PCSP examples is the problem of distinguishing $k$-colorable graphs from those that are not even $l$-colorable for fixed $l \geq k$; see~\cite{Barto2021algebraicPCSP} for further examples. A complete complexity classification for PCSPs seems currently far away. Nevertheless, the algebraic approach via polymorphism  works, and the equivalence of \ref{item:4thmold} and \ref{item:5thmold} in Theorem~\ref{thm:old} also remains valid~\cite{Barto2021algebraicPCSP}.

Finally, the Promise Valued CSP (PVCSP) combines both generalizations. A template is a pair of valued structures of the same signature and the problem is, given   
a sum such as~(\ref{eq:VCSPinst}) and a rational number $\thresh$, to distinguish sums whose minimum computed in $\bfa$ is at most $\thresh$ from those whose minimum in $\bfb$ is greater than $\thresh$. Again, the problem only makes sense if the template satisfies certain properties. An exact characterization of when this happens, Proposition~\ref{prop:FH}, is one of the minor contributions of this paper.

We believe that the PVCSP is an extremely promising research direction for two reasons. First, it is very broad: it includes, for example, all constant factor approximation problems for MaxCSP (both the version where the aim is to approximately maximize the number of satisfied constraints, see Example~\ref{ex:MaxCSP}; and the version where the aim is to approximately minimize the number of unsatisfied constraints). Second, the approach via generalized polymorphisms, so successful in the above special cases, is still available~\cite{Kazda21} (the work is not yet published). 
The only published work on PVCSP that we are aware of is~\cite{ViolaZ21} where the authors, among other results, generalize \ref{item:4thmold} $\iff$ \ref{item:5thmold}  in Theorem~\ref{thm:old} to the PVCSP setting and even consider the more general infinite-domain case.

\subsection{Contributions}

Our main result, Theorem~\ref{th:main}, lifts the equivalence of \ref{item:1thmold}, \ref{item:2thmold}, and  \ref{item:3thmold} in Theorem~\ref{thm:old} to the PVCSP framework.

The generalization of implication \ref{item:1thmold} $\Rightarrow$ \ref{item:2thmold} for connected input valued structures follows easily from the nature of the message passing systems we deal with. General, possibly disconnected input valued structures require an additional argument.\footnote{This subtle issue was not properly handled in~\cite{ButtiD2021distributed}. The present paper thus also fills in a gap in the proof of Theorem~\ref{thm:old}.
}
For the implication \ref{item:2thmold} $\Rightarrow$ \ref{item:3thmold} we employ the approach of~\cite{ButtiD21fractional} and, in a sense, ``decompose'' a solution to the $\sa^1$ relaxation of a PVCSP into three components. One component is a kind of morphism, called here a dual fractional homomorphism, which appeared before in the context of VCSPs with left-hand side (i.e., structural) restrictions~\cite{Carbonnel2022otherside}.\footnote{\cite{Carbonnel2022otherside} uses the terminology ``inverse fractional homomorphism'', however we feel that ``dual'' might better fit the meaning of this concept.}
The decomposition theorem, stated as Theorem~\ref{thm:Decomp}, might be of independent interest.
We also point out that our construction for this decomposition is much simpler than the construction used in~\cite{ButtiD21fractional} for the less general setting. 
The distributed algorithm that we design to prove \ref{item:3thmold} $\Rightarrow$ \ref{item:1thmold} is completely different from the one used for the CSP  in~\cite{ButtiD2021distributed}. The original algorithm relied on a deep theorem from the algebraic CSP theory~\cite{KozikConsistency} about the strength of a certain local propagation algorithm and designed a distributed version of that algorithm. This approach is no longer applicable, even in the (non-valued) PCSP setting. However, we show that a substantially more straightforward and simple idea of directly computing an adjusted form of $\sa^1$ works even in the most general PVCSP framework.

Surprisingly, the implication \ref{item:3thmold} $\Rightarrow$ \ref{item:4thmold} is no longer true for PVCSPs: in Example \ref{ex:SAnotBLP} we present a PVCSP template that is decided by $\sa^1$ but not decided by $\blp$. The converse implication remains valid since $\sa^1$ is a stronger relaxation than $\blp$.

Recall that the equivalence of \ref{item:4thmold} and \ref{item:5thmold} still holds for PVCSPs~\cite{ViolaZ21}; we give a streamlined presentation of the proof using Proposition~\ref{prop:FH}. We also mention, in Example~\ref{ex:pvcsp}, some (P)(V)CSPs that satisfy these conditions, and thus also satisfy the equivalent statements in the main result.

\section{Preliminaries} \label{sec:prelim}

For a tuple $\ba \in A^k$, let $\ba[i]$ denote the $i\ith$ entry of $\ba$. We say that $\textbf{a}$ has a \textit{repetition} if there exist $i \neq j \in [k]$ such that $\ba[i] = \ba[j]$.
We use double curly brackets $\mult{\dots}$ to denote multisets. For a non-negative integer $n$, $n \cdot \mult{\dots}$ stands for the multiset obtained by multiplying the multiplicity of each element in the original multiset by $n$.
Slightly abusing the notation, the set and the multiset of entries of a tuple $\ba$ is denoted by $\{\ba\}$ and $\mult{\ba}$, respectively.

We denote by $\qq$ the set of non-negative rational numbers and by $\qinfty$ the set $\mathbb{Q} \cup \{\infty\}$, where $\infty$ is an additional symbol interpreted as a positive infinity. We set $0 \cdot \infty = 0$ and $c \cdot  \infty = \infty$ for $c>0$.

\subsection{CSP and PCSP}

We present the CSP and PCSP as homomorphism problems. The difference from the presentation in the introduction is merely formal. 

A \textit{signature} $\sigma$ is a finite collection of relation symbols, each with an associated arity. We shall use $\ar(R)$ to denote the arity of a relation symbol $R$. Given a set $A$ and a positive integer $k$, a $k$-ary \textit{relation} on $A$ is a subset of $A^k$. A \textit{(relational) structure} $\bfa$ in the signature $\sigma$, or simply a $\sigma$-structure, consists of a finite set $A$ called the \textit{universe} of $\bfa$, and a relation $\RA$ on $A$ of arity $\ar(R)$ for each $R \in \sigma$. Notice that the universe of every structure in this paper is assumed to be finite. Two structures are \emph{similar} if they have the same signature.

Let $\bfi$, $\bfa$ be $\sigma$-structures. A \textit{homomorphism} from $\bfi$ to $\bfa$ is a map $h: I \to A$ such that for every $R \in \sigma$ and every tuple $\mb{v} \in \RI$ it holds that $h(\mb{v}) \in \RA$, where $h$ is applied to $\mb{v}$ component-wise.
If there exists a homomorphism from $\bfi$ to $\bfa$ we say that $\bfi$ is homomorphic to $\bfa$.

For a relational $\sigma$-structure $\bfa$, the \emph{CSP over $\bfa$}, denoted $\csp(\bfa)$, is the problem of deciding whether an input $\sigma$-structure $\bfi$ is homomorphic to $\bfa$. The structure $\bfa$ is also referred to as a \emph{template} in this context. The translation of the presented definition of $\csp(\bfa)$ to the formalism used in the introduction is given by defining the set of \emph{constraints} $\mathcal{C}_\bfi$ as the set of formal expressions of the form $R(\mb{v})$ where $R \in \sigma$ and $\mb{v} \in \RI$.

Given two $\sigma$-structures $\bfa$ and $\bfb$,
the \emph{Promise CSP over $(\bfa,\bfb)$}, denoted $\pcsp(\bfa,\bfb)$,  is defined as follows: given a  $\sigma$-structure $\bfi$, output $\yes$ if $\bfi$ is homomorphic to $\bfa$, and output $\no$ if $\bfi$ is not homomorphic to $\bfb$.%
\footnote{
We do not impose any requirements on the algorithm in the case that $\bfi$ is neither a $\yes$ instance nor a $\no$ instance. Alternatively, we are \emph{promised} that the input is a $\yes$ instance or a $\no$ instance. 
} This problem makes sense iff the sets of $\yes$ and $\no$ instances are disjoint. It is easy to see that this happens exactly when $\bfa$ is homomorphic to $\bfb$. Such pairs of structures $(\bfa,\bfb)$ are called \emph{PCSP templates}.

\subsection{PVCSP} 
We formalize PVCSPs in a similar way to PCSPs. The difference from the presentation in the introduction is slightly more substantial, as we shall briefly discuss later.

A $k$-ary \textit{valued relation} on $A$ is a function $R:A^k \to \qinfty$. A \textit{valued $\sigma$-structure} $\bfa$  consists of a finite universe $A$, together with a valued relation $\RA$ of arity $\ar(R)$ on $A$ for each $R \in \sigma$.
Valued structures are sometimes referred to as \textit{general-valued} in the literature \cite{KolmogorovTZ15,ThapperZ17} to emphasize that relations in $\bfa$ may take non-finite values. A $\sigma$-structure $\bfa$ is said to be \textit{non-negative finite-valued} if for every $R \in \sigma$, the range of $\RA$ is contained in $\qq$.

Let $\bfi$, $\bfa$ be valued $\sigma$-structures, where $\bfi$ is non-negative finite-valued. The \emph{value} of a map $h: I \to A$ for $(\bfi,\bfa)$, and the \emph{optimum value} for $(\bfi, \bfa)$ are given by 
\begin{equation*}
\val(\bfi,\bfa,h) = \sum_{R \in \sigma} \sum_{\textbf{v} \in I^{\ar(R)}} R^\bfi(\textbf{v}) R^\bfa(h(\textbf{v})), \quad\quad
\opt(\bfi,\bfa)= \min_{h:I \to A}\val(\bfi,\bfa,h).
\end{equation*}
 For two valued $\sigma$-structures $\bfa$ and $\bfb$, the \textit{Promise Valued CSP over $\ab$}~\cite{Kazda21,ViolaZ21}, denoted $\pvcsp \ab$, is defined as follows: given a pair $(\bfi,\thresh)$, where $\bfi$ is a non-negative finite-valued $\sigma$-structure  and $\thresh \in \mathbb{Q}$ is a \emph{threshold}, output $\yes$ if $\opt\ia \leq \thresh$, and output $\no$ if $\opt\ib > \thresh$. 
 We call $(\bfa,\bfb)$ a \emph{PVCSP template} if the sets of $\yes$ and $\no$ instances are disjoint. We show in Proposition~\ref{prop:FH} that this least restrictive meaningful requirement on a PVCSP template coincides with the choice taken in~\cite{ViolaZ21}.
 
 Notice that the values of $R$ have a different intended meaning in the template valued structures $\bfa$, $\bfb$ and in the input valued structure $\bfi$. For the template, $\RA(\ba)$ and $\RB(\bb)$ should be understood as the \emph{cost} of $\ba$ and $\bb$: we wish an assignment $h$ to map tuples of variables to tuples of domain elements that are as cheap as possible (and, in fact, $\RA$ or $\RB$ is often referred to as a \emph{cost function}). On the other hand, $\RI(\bv)$ is the \emph{weight} of the tuple of variables $\bv$: we need to be more concerned about heavy tuples of variables, while we may ignore the tuples of zero weight (recall that $0 \cdot \infty = 0$).
 As an example, observe that the PCSP over a pair of structures $(\bfa',\bfb')$ is essentially the same problem as the PVCSP over the pair of $\{0,\infty\}$-valued structures $\ab$, where tuples in the latter template are given zero cost iff they belong to the corresponding relations in the former template; while to an instance $\bfi'$ of the PCSP corresponds a non-negative finite-valued structure $\bfi$ where the cost of a tuple is zero iff the tuple does \emph{not} belong to the corresponding relation in $\bfi'$ (and costs of the remaining tuples are arbitrary positive rationals), together with any threshold $\thresh \in \qq$.
 
 For a PVCSP input valued $\sigma$-structure $\bfi$ we define the set of \emph{constraints} $\mathcal{C}_\bfi$ as the set of formal expressions of the form $R(\bv)$ where $R \in \sigma$, $\bv \in I^{\ar(R)}$, and $\RI(\bv) > 0$; the value $\RI(\bv)$ is the \emph{weight} of the constraint.
 This almost translates the presented definition of PVCSP to the version from the introduction: weights of constraints can be emulated by repeating constraints in (\ref{eq:VCSPinst}) (and modifying the threshold $\thresh$ if necessary). 
 However, the repetition can cause an exponential blow up of the instance size. Nevertheless, this difference between the two formalisms is inessential for our purposes.

 We say that a valued relation $R^\bfi$ has no repetitions if $R^\bfi(\bv)=0$ whenever $\bv$ has a repetition. Similarly, we say that an input valued structure $\bfi$ has no repetitions if none of its valued relations has a repetition.

\begin{example} \label{ex:MaxCSP}
As mentioned in the introduction, the PVCSP framework can be used to model constant factor approximation problems for MaxCSP. More concretely, suppose that we want to find a $c$-approximation for $\csp(\bfa)$ for some (non-valued) $\sigma$-structure $\bfa$ and some $c < 1$. One can model this problem as $\pvcsp(\bfa',\bfb')$ where $A'=B'=A$ and for all $R \in \sigma$ and $\ba \in A^{\ar(R)}$,  $R^{\bfa'}(\ba)=-1$ if $\ba \in \RA$ and $R^{\bfa'}(\ba)=0$ otherwise; and $R^{\bfb'}(\ba)=\frac{1}{c} R^{\bfa'}(\ba)$. 
Given an instance $\bfi$ of $\csp(\bfa)$ and a parameter $0 < \beta \leq 1$, we turn it into an instance $(\bfi',-\beta m)$ of $\pvcsp(\bfa',\bfb')$ in a natural way, where $\bfi'$ is a 0-1 valued structure and $m$ is the number of constraints in $\bfi'$. Then, $\opt(\bfi',\bfa')\leq -\beta m$ if a $\beta$-fraction of all constraints of $\bfi$ can be satisfied in $\bfa$, and $\opt(\bfi',\bfb')>- \beta m$ if not even a $c\beta$-fraction of the constraints of $\bfi$ can be satisfied in $\bfa$.
\end{example}

\subsection{Linear programming relaxations} 

Given two valued $\sigma$-structures $\bfi$ and $\bfa$ where $\bfi$ is non-negative finite-valued, the systems of inequalities $\blp\ia$ and $\sa^1 \ia$ contain a variable $p_v(a)$ for every $v \in I$ and every $a \in A$, and a variable $p_{R(\bv)}(\ba)$ for every $R(\bv) \in \mathcal{C}_\bfi$ and every $\ba \in A^{\ar(R)}$.  $\blp\ia$ is the following linear program.

\begin{empheq}[box=\fbox]{align}
& \qquad \mathrlap{ \opt^{\blp} \ia := \min \sum_{R(\bv) \in \mathcal{C}_\bfi}  \sum_{\textbf{a} \in A^{\ar(R)}} p_{R(\bv)}(\textbf{a}) R^\bfi(\textbf{v}) R^\bfa(\textbf{a})} \label{eq:objBLP} \tag{$\star$}\\
\quad  &\textnormal{subject to:}\nonumber && \\
& \qquad p_{R(\bv)}(\ba) \geq 0 && R(\bv) \in \mathcal{C}_\bfi, \ \ba \in A^{\ar(R)} \label{eq:SA0} \\
& \qquad \sum_{a \in A} p_v(a)=1  &&  v \in I \label{eq:SA1}\\
& \qquad   p_v(a) = \sum_{\mathclap{\ba \in A^{\ar(R)}, \ba[i]=a}} \, p_{R(\textbf{v})}(\ba) &&  a\in A, \ R(\bv) \in \mathcal{C}_\bfi, \ 
    i \in [\ar(R)] \mbox{ s.t. } \bv[i]=v   \label{eq:SA2}\\
& \qquad  p_{R(\textbf{v})}(\ba) = 0 &&  R(\bv) \in \mathcal{C}_\bfi, \  \ba \in A^{\ar(R)} \textnormal{ s.t. } R^\bfa(\ba) = \infty \quad \label{eq:SA3}
\end{empheq}

As for the program $\saia$, the objective function, denoted $\opt^{\sa^1}\ia$, is given by the same objective function as in $\blp\ia$. The variables are subject to all the constraints in $\blp\ia$, but in addition, they are also subject to the following constraint.

\begin{empheq}[box=\fbox]{align}
  & \quad \quad p_{R(\textbf{v})}(\ba) = 0   \qquad \qquad  &&  R(\bv) \in \mathcal{C}_\bfi, \ \ba \in A^{\ar(R)} \label{eq:SA4} \\
&&& \exists i,j \in [\ar(R)] \textnormal{ such that } \bv[i] = \bv[j] \textnormal{ and }\ba[i] \neq \ba[j] \nonumber
\end{empheq}  
Notice that in general $\opt^\blp\ia \leq \opt^{\sa^1}\ia$. Moreover, in the particular case where $\bfi$ has no repetitions, $\blp$ and $\sa^1$ are the same linear program and so $\opt^\blp\ia = \opt^{\sa^1}\ia$.

For a linear program $\textnormal{L} \in \{\blp,\sa^1\}$ we say that $\textnormal{L}\ia$ is \textit{feasible} if there exists a rational solution to the system $\textnormal{L}\ia$. Notice that then (\ref{eq:objBLP}) makes sense since $\RA(\ba)=\infty$ implies $p_{R(\bv)}=0$ and $0\cdot\infty = 0$ (formally, one should skip these summands in (\ref{eq:objBLP})).
If the linear program is infeasible, then we set $\opt^{\textnormal{L}}\ia = \infty$.

The LP constraints (\ref{eq:SA0})--(\ref{eq:SA2}) ensure that, for each $R(\bv) \in \mathcal{C}_\bfi$, the values of $p_{R(\bv)}(\bf a)$ form a probability distribution on $A^{\ar(R)}$ (which is additionally consistent with $p_v(a)$'s). The inner sum in (\ref{eq:objBLP}) is equal to the expected ``cost'' of the constraint $R(\bv)$ with weight $\RI(\bv)$ when $\bv$ is evaluated according to this distribution. From this observation it is apparent that $\opt^{\textnormal{L}}\ia \leq \opt \ia$. 
We say that $\textnormal{L}$ \textit{decides} $\pvcsp \ab$ if, for every input structure $\bfi$, we have $\opt \ib \leq \opt^{\textnormal{L}}\ia$. Note that in such a case the algorithm for $\pvcsp \ab$ that answers $\yes$ iff $\opt^{\textnormal{L}}\ia \leq \thresh$ (where $\thresh$ is the input threshold) is correct, so the definition makes sense.%
\footnote{
We remark that in \cite{ButtiD21fractional}, the feasibility of the program $\sa^1\ia$ was alternatively phrased as the existence of a ``fractional homomorphism'' from $\bfi$ to $\bfa$, to stress that the linear system $\sa^1\ia$ is a (fractional) relaxation of homomorphism in the same way as the equivalence relation $\eqone$ defined below is a relaxation of isomorphism. Nonetheless, in this paper we avoid this terminology as it clashes with the notion of fractional homomorphism defined in Section \ref{sec:FH} as a unary fractional polymorphism. 
}

\subsection{Polymorphisms}

An $n$-ary \textit{polymorphism} of a pair of similar structures $\ab$ is an $n$-ary operation $f:A^n \to B$ such that for every relation symbol $R$ in the signature of $\bfa$ and $\bfb$, the coordinate-wise application of $f$ to any list of $n$ tuples from $R^\bfa$ results in a tuple in $R^\bfb$. Note that a unary polymorphism of $\ab$ is just a homomorphism from $\bfa$ to $\bfb$.
An $n$-ary operation $f:A^n \to B$ is said to be \textit{symmetric} if for every $a_1,\ldots,a_n \in A$ and every permutation $\rho$ on $[n]$ we have that $f(a_1,\ldots,a_n) = f(a_{\rho(1)},\ldots,a_{\rho(n)})$.

An $n$-ary \textit{fractional polymorphism} \cite{ViolaZ21} of two valued $\sigma$-structures $\ab$ is a probability distribution $\omega$ on the set $B^{A^n}:=\{f:A^n \to B\}$ such that for every $R \in \sigma$ and every list of $n$ tuples  $\textbf{a}_1,\ldots,\textbf{a}_n \in A^{\ar(R)}$ we have that
\[\sum_{f \in B^{A^n}} \omega(f) R^\bfb(f(\textbf{a}_1,\ldots,\textbf{a}_n)) \leq \frac{1}{n}\sum_{i=1}^n  R^\bfa(\textbf{a}_i)\]
where $f$ is applied to $\textbf{a}_1,\ldots,\textbf{a}_n \in A^{\ar(R)}$ component-wise.\footnote{We use here a simpler concept than fractional polymorphism as defined in~\cite{ViolaZ21}, which will be sufficient for our purposes.
}

The \textit{support} of $\omega$ is the set of functions $f:A^n \to B$ such that $\omega(f) >0$. We say that $\omega$ is \textit{symmetric} if every operation in its support if symmetric. 

The following theorem was proved in~\cite{ViolaZ21}; we provide a somewhat streamlined argument in the spirit of~\cite{Barto2021algebraicPCSP} in Section \ref{sec:FH}.

\begin{theorem}  \label{thm:BLPandSym}
Let $\ab$ be a promise valued template of signature $\sigma$. Then the following are equivalent.
\begin{enumerate}[label=\textnormal{(\roman*)}]
\setcounter{enumi}{3}
    \item $\blp$ decides $\pvcsp\ab$; \label{item:4thmain}
    \item $\ab$ has symmetric fractional polymorphisms of every arity. \label{item:5thmain}    
\end{enumerate}
\end{theorem}

\begin{example} \label{ex:pvcsp}
A CSP that can be decided by $\blp$ is e.g. the Horn-3-Sat, where the template has domain $\{\mathrm{true}, \mathrm{false}\}$ and two relations defined by $\neg x \vee \neg y \vee \neg z$ and $\neg x \vee \neg y \vee z$.
A well-known class of templates with $\blp$-decidable VCSPs are those that contain only submodular valued relations (see~\cite{KZ_vcsps}). Finally, the 2-approximation of the Vertex Cover problem~\cite{KZ_vcsps} is a PVCSP decidable by $\blp$. In all the mentioned examples, it is not hard to find symmetric (fractional) polymorphisms of every arity.
\end{example}

\subsection{Graph of an input, iterated degree, distributed model}

We represent an input $\sigma$-structure $\bfi$ to a PVCSP as a labeled bipartite graph $\graph{\bfi}$, also known as the \textit{factor graph} of $\bfi$ in the non-valued setting \cite{Fioretto2018}. This representation allows us to define iterated degrees of variables and constraints as well as our distributed model.

$\graph{\bfi}$ has one vertex for each constraint $R(\bv) \in \mathcal{C}_\bfi$, labeled $(R,q)$ where $q=R^{\bfi}(\bv)$ ($>0$), and one vertex for each variable, with empty label. Vertex $v \in I$ is adjacent to a vertex $R(\bv) \in \mathcal{C}_\bfi$ if $v \in \{\bv\}$; the edge is labeled $S = \{i: \bv[i]=v\}$. The label of a vertex $x$ is denoted $\lbl_x$, the label of an edge $\{x,y\}$ is denoted $\lbl_{\{x,y\}}$.

We call $\bfi$ \emph{connected} if $\graph{\bfi}$ is. Similarly, we say that $\bfi'$ is a connected component of $\bfi$ if $\graph{\bfi'}$ is a connected component of $\graph{\bfi}$. 

The $k\ith$ \emph{iterated degree} of a vertex $x$, where $x$ is a variable or a constraint, is defined inductively by $\delta_0^{\bfi}(x) = \lbl_x$, and $\delta_{k+1}^{\bfi}(x) = \mult{(\lbl_{\{x,y\}},\delta_k^{\bfi}(y)) \mid y \mbox{ is adjacent to } x \textnormal{ in } \graph{\bfi}}$. The \emph{iterated degree} of a vertex $x$ is defined  as  $\delta^{\bfi}(x) = (\delta_0^{\bfi}(x),\delta_1^{\bfi}(x),\delta_2^{\bfi}(x), \ldots)$. For vertices $x$ and $y$ we write $x \eqone y$ if they have the same iterated degrees. 
Note that the iterated degrees are analogues of colors in the 1-dimensional Weisfeiler-Leman color refinement algorithm~\cite{leman1968reduction} for graph isomorphism test.
The \emph{iterated degree sequence} of $\bfi$ is defined as
$\delta(\bfi)=\mult{\delta^{\bfi}(x) \mid x \in I \cup \mathcal{C}_\bfi}$; for two $\sigma$-structures $\bfi$, $\bfj$, we write $\bfi \eqone \bfj$ if they have the same iterated degrees sequence.%
\footnote{The degree sequence is often defined to be a list. However, when looking at iterated degree it is common~\cite{Ramana1994,Scheinerman2011fractional} and more practical to use multisets instead of lists, while maintaining the terminology \textit{sequence} to highlight that we are dealing with a generalisation of the classical concept of degree sequence.} 
Notice that in order to prove that $\bfi \eqone \bfj$ it is sufficient to show that $\mult{\delta^{\bfi}(x) \mid x \in I }=\mult{\delta^{\bfj}(x) \mid x \in J}$.

The computational model for solving $\pvcsp \ab$ on a distributed network is as follows.
An input valued structure $\bfi$ is represented as a bipartite message passing network designed as $\graph{\bfi}$: we have an agent $\alpha(x)$ for every vertex $x \in I \cup \mathcal{C}_\bfi$ and the communication channels exactly correspond to edges in $\graph{\bfi}$ and have the same labels. Every agent in the network knows only the template, the threshold, the number of variables ($|I|$), the number of constraints ($|\mathcal{C}_\bfi|$), and the labels of their controlled variable and of the adjacent channels.
The agents are  anonymous, they all run the same deterministic algorithm, and the communication proceeds in synchronous rounds. 
 For a more detailed discussion on the distributed set-up, we refer the reader to \cite{ButtiD2021distributed}.
 
We say that a distributed algorithm solves an instance $(\bfi,\thresh)$ of  $\pvcsp\ab$ if the algorithm terminates and the terminating state of every process is $\yes$ if $(\bfi,\thresh)$ is a $\yes$ instance of $\pvcsp\ab$, and $\no$ if $(\bfi,\thresh)$ is a $\no$ instance of $\pvcsp\ab$. We say that a distributed algorithm solves $\pvcsp\ab$ if it solves every connected instance of $\pvcsp\ab$ (note here that it makes little sense to run a distributed algorithm on a disconnected network).

\section{Fractional homomorphisms and $\sa^1$} \label{sec:FH}

We start by stating the characterization of PVCSP templates in terms of fractional homomorphisms. The result will also be useful in the proof of Theorem~\ref{thm:BLPandSym}.

A \textit{fractional homomorphism} \cite{ThapperZ12,ViolaZ21} from $\bfa$ to $\bfb$  is a unary fractional polymorphism of $\ab$, or equivalently, a probability distribution $\mu$ over $B^{A}$ such that for every $R \in \sigma$ and every $\textbf{a} \in A^{\ar(R)}$ we have that
\begin{equation} \label{eq:FH}\sum_{f \in B^{A}} \mu(f) R^\bfb(f(\textbf{a})) \leq R^\bfa(\textbf{a}).
\end{equation}

If there exists a fractional homomorphism from $\bfa$ to $\bfb$, we say that $\bfa$ is  fractionally homomorphic to $\bfb$ and we write $\bfa \to_f \bfb$. 

The implication (\ref{item:1FH}) $\Rightarrow$ (\ref{item:2FH}) in the following proposition is a well-known and easy calculation~(see e.g.~\cite{ThapperZ12}). The converse implication appears to be new, although the proof technique via Farkas' Lemma~\cite{farkas} is standard in the VCSP area. 

\begin{proposition} \label{prop:FH}
For any two valued $\sigma$-structures $\bfa$ and $\bfb$, the following are equivalent.
\begin{enumerate}
    \item There exists a fractional homomorphism from $\bfa$ to $\bfb$. \label{item:1FH}
    \item For all non-negative finite-valued $\sigma$-structures $\bfi$, $\opt(\bfi,\bfb) \leq \opt(\bfi,\bfa)$. \label{item:2FH}
    \end{enumerate}
\end{proposition}

\begin{proof}
(\ref{item:1FH}) $\Rightarrow$ (\ref{item:2FH}) 
Let $\mu$ be a fractional homomorphism from $\bfa$ to $\bfb$, let $g:I \to A$ be such that $\opt(\bfi,\bfa) = \val(\bfi,\bfa,g)$, and let $f \in B^A$ be some map that minimizes $\val(\bfi,\bfb,f\circ g)$. Then \begin{align*}
    \opt(\bfi,\bfb) & \leq \val(\bfi,\bfb,f\circ g)  \leq \sum_{f' \in B^A} \mu(f') \val(\bfi,\bfb,f'\circ g) \\
    & = \sum_{R \in \sigma} \sum_{\textbf{v} \in I^{\ar(R)}} R^\bfi(\bv) \sum_{f' \in B^A} \mu(f') R^\bfb(f' \circ g(\bv))\\
   & \leq \sum_{R \in \sigma} \sum_{\bv \in I^{\ar(R)}} R^\bfi(\bv) R^\bfa(g(\bv)) = \val(\bfi,\bfa,g) = \opt(\bfi,\bfa).
\end{align*}

$(\ref{item:2FH}) \Rightarrow (\ref{item:1FH})$. The idea for this proof is to assume that there is no fractional homomorphism from $\bfa$ to $\bfb$, formulate this fact as infeasibility of a system of linear inequalities, and then use a version of Farkas' Lemma to find $\bfi$ with $\opt\ib > \opt \ia$.  

The existence of a fractional homomorphism from $\bfa$ to $\bfb$  can be reformulated as the following system of linear inequalities, where there is a rational-valued variable $\mu_f$ for every $f \in B^A$.

\begin{empheq}[box=\fbox]{align}
   \qquad  \textnormal{variables: }& \quad \mu_f  \textnormal{ for all } f \in B^A\nonumber\\
   \qquad  \textnormal{constraints: }& \quad  \sum_{f \in B^A} \mu_f R^{\bfb}(f(\textbf{a})) \leq R^{\bfa}(\textbf{a}) \textnormal{ for all } R \in \sigma \textnormal{ and } \textbf{a} \in A^{\ar(R)} \label{eq:fh1} \qquad\\ 
& \quad \sum_{f \in B^A} \mu_f \geq 1 \nonumber\\
& \quad \mu_f \geq 0 \textnormal{ for all } f \in B^A. \nonumber
\end{empheq}
\\

If there is no fractional homomorphism from $\bfa$ to $\bfb$, this system is  infeasible. 

We now deal with infinite coefficients.
Define $B^A_{<\infty}= \{f \in B^A: \forall R \in \sigma, \forall \ba \in A^{\ar(R)} \mathbin{,} \  \RA(\ba) < \infty \textnormal{ implies } \RB(f(\ba))<\infty\}$.
Now consider the new linear system obtained from the above by first removing all the inequalities in (\ref{eq:fh1}) where $\RA(\ba)=\infty$ (since these inequalities are always satisfied), and second, by removing the variable $\mu_f$ for all $f \in B^A \setminus B^A_{<\infty}$ and changing (\ref{eq:fh1}) so that the sums run over  $B^A_{<\infty}$ only (since we need to have $\mu_f=0$ for $f \in B^A \setminus B^A_{<\infty}$ in any feasible solution). Clearly, the system of linear inequalities resulting from this procedure remains infeasible and does not contain infinite coefficients.

This system of linear inequalities can be rewritten in matrix form as $M \textbf{f} \leq \textbf{a}$ subject to $\textbf{f} \geq 0$, where $\textbf{f} \in \mathbb{Q}_{\geq}^{ B^A_{<\infty}}$ is the vector of unknowns, and $M$ is a real-valued matrix. By Farkas' Lemma, the system of inequalities $M^T \textbf{y} \geq 0$ subject to $\textbf{a}^T \textbf{y} < 0$ and $\textbf{y} \geq 0$ is feasible. Explicitly, the latter system is the following.

\begin{empheq}[box=\fbox]{align} \label{eq:BTy0}
    \qquad \textnormal{variables: }& \quad y, x_{R,\ba} \textnormal{ for every } R \in \sigma \textnormal{ and }\textbf{a} \in A^{\ar(R)} \nonumber \textnormal{ with } \RA(\ba)< \infty\\
    \qquad \textnormal{constraints: }& \quad  \sum_{R \in \sigma}\sum_{\substack{\textbf{a} \in A^{\ar(R)}\\ \RA(\ba)< \infty}} x_{R,\ba} R^{\bfb}(f(\textbf{a})) \geq y  \textnormal{ for all } f \in  B^A_{<\infty} \qquad \\ 
    & \quad \sum_{R \in \sigma} \sum_{\substack{\textbf{a} \in A^{\ar(R)}\\ \RA(\ba)< \infty}} x_{R,\ba} R^{\bfa}(\textbf{a}) < y \nonumber\\
    & \quad  x_{R,\ba} \geq 0 \textnormal{ for all } R \in \sigma , \textbf{a} \in A^{\ar(R)} \nonumber\\
    & \quad  y \geq 0. \nonumber
\end{empheq}

Eliminating $y$, and adding trivially satisfied constraints to (\ref{eq:BTy0}) for all $f \in B^A \setminus B^A_{<\infty}$, we get that the following system is feasible.

\begin{empheq}[box=\fbox]{align} \label{eq:BTy0reduced}
    \textnormal{variables: }& \quad  x_{R,\ba} \textnormal{ for every } R \in \sigma \textnormal{ and }\textbf{a} \in A^{\ar(R)} \nonumber  \textnormal{ with } \RA(\ba)< \infty\\
    \textnormal{constraints: }& \quad \sum_{R \in \sigma} \sum_{\substack{\textbf{a} \in A^{\ar(R)}\\ \RA(\ba)< \infty}} x_{R,\ba} R^{\bfb}(f(\textbf{a})) >  \sum_{R \in \sigma} \sum_{\substack{\textbf{a} \in A^{\ar(R)}\\ \RA(\ba)< \infty}} x_{R,\ba} R^{\bfa}(\textbf{a}) \textnormal{ for all } f \in  B^A \\ 
    & \quad x_{R,\ba} \geq 0 \textnormal{ for all } R \in \sigma , \textbf{a} \in A^{\ar(R)} \nonumber
\end{empheq}

Let $x_{R,\ba}$ for $R \in \sigma$, $\ba \in A^r$ be a feasible solution to (\ref{eq:BTy0reduced}), and consider the structure $\bfi$ with domain $I=A$ and relations given by $R^\textbf{I}(\textbf{a})=x_{R,\ba}$ for $\textbf{a} \in A^{\ar(R)}$ with $\RA(\ba)< \infty$ and $R^\textbf{I}(\textbf{a})=0$ whenever $\RA(\ba)= \infty$. Notice that $\bfi$ is non-negative finite-valued, that the right-hand side in the first inequality is equal to $\val(\bfi,\bfa,\mathrm{id})$, (where $\mathrm{id}$ denotes the identity function) and that the left-hand side is equal to $\val(\bfi,\bfb,f)$. 
Therefore $\opt\ib > \opt \ia$, as required.
\end{proof}

\begin{proof}[Sketch of proof of Theorem~\ref{thm:BLPandSym}]

For an integer $m \geq 1$, let $\lpam$ be the structure whose universe consists of $A$-multisets of size $m$ and whose valued relations are defined by the following formula where $R \in \sigma$ and $s_1, \ldots, s_r$ are from the universe.
 \[R^{\lpam}(s_1,\ldots,s_{r}):=\frac{1}{m} \min_{\substack{\textbf{t}_1,\ldots,\textbf{t}_r \in A^m\\ \mult{\textbf{t}_i}=s_i}} \sum_{i=1}^m R^{\bfa}(\textbf{t}_1[i],\ldots,\textbf{t}_r[i]).\]
Variants of such structures  have been defined in the literature both for (P)CSP \cite{Kun2012,Barto2021algebraicPCSP} and for VCSP \cite{ThapperZ12, ViolaZ21}. These papers also explicitly or implicitly observe the following properties. 

\begin{enumerate}
    \item $\opt^{\blp}\ia = \min_{m \geq 1}\opt(\bfi,\lpam)$ for all non-negative finite-valued $\bfi$. \label{item:lpam1}
    \item For all $m\geq 1$, $\lpam \to_f \bfb$ if and only if $\ab$ has an $m$-ary symmetric fractional polymorphism. \label{item:lpam2}
\end{enumerate}

The proof can be now finished using Proposition~\ref{prop:FH}. For \ref{item:4thmain} $\Rightarrow$ \ref{item:5thmain} suppose that $\ab$ does not have a symmetric polymorphism of some arity $m$. Then, there is no fractional homomorphism from $\lpam$ to $\bfb$. It follows from Proposition \ref{prop:FH} that there exists some structure $\bfi$ such that $\opt(\bfi,\bfb) > \opt(\bfi,\lpam) \geq \opt^{\blp}\ia$. Hence, $\blp$ does not decide $\pvcsp\ab$. 
On the other hand, for \ref{item:5thmain} $\Rightarrow$ \ref{item:4thmain}, assume that $\ab$ has symmetric fractional polymorphisms of every arity. 
Let $m \geq 1$ be such that $\opt(\bfi,\lpam)$ is minimal. We know that $\lpam$ is fractionally homomorphic to $\bfb$ and therefore for all finite-valued structures $\bfi$, $\opt(\bfi,\bfb) \leq \opt(\bfi,\lpam)=\opt^{\blp}\ia$. Hence, $\blp$ decides $\pvcsp\ab$.
\end{proof}

The decomposition theorem mentioned in the introduction uses a concept that is ``dual'' to  fractional homomorphism, as suggested by the following Proposition~\ref{prop:dualFH}. Here we only present the proof of the implication that is needed for the decomposition theorem. The proof of the other implication uses techniques similar to the ones deployed in Proposition \ref{prop:FH}, and we refer the reader to \cite{Carbonnel2022otherside} for the details.

We define a \textit{dual fractional homomorphism} from $\bfi$ to $\bfj$ ($\bfi \to_{df} \bfj$) to be a probability distribution $\eta$ over $J^I$ such that for every $R \in \sigma$ and every $\bu \in J^{\ar(R)}$ we have that
\begin{equation} \label{eq:dualFH}
 R^{\bfj}(\bu) \geq \sum_{f \in J^I} \eta(f) \sum_{\substack{\bv \in I^{\ar(R)} \\ \bu = f(\bv)}} R^{\bfi}(\bv).
\end{equation}
\begin{proposition} \label{prop:dualFH}
For any two non-negative finite-valued $\sigma$-structures $\bfi$ and $\bfj$, the following are equivalent.
\begin{enumerate}
    \item There exists a dual fractional homomorphism from $\bfi$ to $\bfj$. \label{item:1dualFH}
    \item For all valued $\sigma$-structures $\bfa$, $\opt(\bfi,\bfa) \leq \opt(\bfj,\bfa)$. \label{item:2dualFH}
    \end{enumerate}
\end{proposition}

\begin{proof}
$(\ref{item:1dualFH}) \Rightarrow (\ref{item:2dualFH})$. Let $\eta$ be a dual fractional homomorphism from $\bfi$ to $\bfj$, and $g:J \to A$ be such that $\opt(\bfj,\bfa) = \val(\bfj,\bfa,g)$. Then
\begin{align*}
    \opt(\bfj,&\bfa)  =  \sum_{R \in \sigma} \sum_{\bu \in J^{\ar(R)}} R^\bfj(\bu) R^\bfa(g(\bu)) \\
     &\geq \sum_{R \in \sigma} \sum_{\bu \in J^{\ar(R)}} \sum_{f \in J^I} \eta(f) \sum_{\substack{\bv \in I^{\ar(R)} \\ \bu = f(\bv)}} R^{\bfi}(\bv)  R^\bfa(g \circ f(\bv))  
     = \sum_{f \in J^I} \eta(f) \val(\bfi,\bfa,g \circ f),
\end{align*}
which implies that there exists some function $f':I \to J$ such that $ \val(\bfi,\bfa,g \circ f') \leq  \opt(\bfj,\bfa) $, hence $ \opt(\bfi,\bfa) \leq  \opt(\bfj,\bfa)$ as required. Notice that this holds regardless of whether $\bfa$ is finite-valued or general-valued. 
\end{proof}

\section{The decomposition theorem}

In this section we state and prove the decomposition theorem. This provides a connection between the combinatorial and the LP-based characterizations of the class of $\pvcsp$ templates that are the subject of our main result, and thus is a fundamental step in the proof of Theorem \ref{th:main}, namely the implication \ref{item:2thmain} $\Rightarrow$ \ref{item:3thmain}.
We refer to~\cite{ButtiD21fractional} for a more detailed discussion about (a weaker form of) this result.

\begin{theorem} \label{thm:Decomp}
Let $\bfi$, $\bfa$ be a pair of similar valued structures, where $\bfi$ is non-negative and finite-valued. 
Then there exist non-negative finite-valued structures $\bfy_1,\bfy_2$ such that
\begin{enumerate}
    \item $\bfi \to_{df} \bfy_1$,
    \item $\bfy_1 \eqone \bfy_2$, and
    \item $\opt(\bfy_2,\bfa)\leq \opt^{\sa^1}\ia$.
\end{enumerate}
\end{theorem}

\begin{proof}
If $\saia$ is not feasible, then we can take $\bfy_1=\bfy_2=\bfi$, and the statement follows trivially, so from now on we shall assume that $\saia$ is feasible.
Let  $p_v(a)$, $p_{R(\bv)}(\ba)$ form an optimal solution of $\saia$ and let $m>0$ be an integer such that all the values $mp_v(a)$ and $mp_{R(\bv)}(\ba)$ are integers. Note that these integers are non-negative by (\ref{eq:SA0}) and  (\ref{eq:SA2}).

We define the universe of both valued structures $\bfy_1$ and $\bfy_2$ as $Y_1=Y_2 = [m] \times I$. The valued structure $\bfy_1$ is simply a ``scaled disjoint union'' of $m$ copies of $\bfi$: we set $R^{\bfy_1}((k,\bv[1])\mathbin{,} (k,\bv[2])\mathbin{,} \dots\mathbin{,} (k,\bv[\ar(R)]))=1/m \cdot R^{\bfi}(\bv)$ for every $k \in [m]$, $\bv \in I^{\ar(R)}$, and the weight of the remaining tuples is set to 0. Observe that $\bfi \to_{df} \bfy_1$ by the dual fractional homomorphism given by the uniform distribution over $f_k$, $k \in [m]$, where $f_k: I \to Y_1$ is defined by $f_k(v) = (k,v)$ for all $v \in I$. Also notice that the iterated degree of each $(k,v)$ is obtained from the iterated degree of $v$ by scaling down each constraint label $(R,q)$ to $(R,q/m)$. 

The structure $\bfy_2$ is a ``twisted'' version of $\bfy_1$ (the construction is a version of the twisted product  from~\cite{Kun13}).
For every $v \in I$, fix a tuple $\bp_v \in A^m$ in which $a \in A$ appears exactly $mp_v(a)$ times -- note that this is possible since the $mp_v(a)$ sum up to $m$ by (\ref{eq:SA1}). We define $h: Y_2 \to A$ by $h(k,v) = \bp_v[k]$ for all $k \in [m]$ and $v \in I$. The structure $\bfy_2$ is constructed so that the value of $h$ for $(\bfy_2,\bfa)$ is $\opt^{\sa^1}\ia$, as follows. For every $R(\bv) \in \mathcal{C}_\bfi$, denote $r=\ar(R)$, and consider an $m \times r$ matrix $Q$ that has, for each $\ba \in A^{r}$, exactly $mp_{R(\bv)}(\ba)$ rows equal to $\ba$. Note that the $i\ith$ column contains $a \in A$ exactly $mp_{\bv[i]}(a)$ times by (\ref{eq:SA2}), in other words, the multiset of elements of this columns is equal to $\mult{\bp_{\bv[i]}}$; in particular, $Q$ indeed has $m$ rows. Moreover, if $\bv[i] = \bv[j]$, then the columns $i$ and $j$ are identical by (\ref{eq:SA4}). It follows that there are permutations $\rho_1, \dots, \rho_r: [m] \to [m]$ such that 
\begin{itemize}
  \item for every $k \in [m]$, the $k\ith$ row of $Q$ is equal to 
  $(\bp_{\bv[1]}[\rho_1(k)], \bp_{\bv[2]}[\rho_2(k)], \dots, \bp_{\bv[r]}[\rho_r(k)])$; 
  \item for every $i,j \in [r]$, if $\bv[i]=\bv[j]$ then $\rho_i=\rho_j$.
\end{itemize}
We set $R^{\bfy_2}((\rho_1(k),\bv[1]), (\rho_2(k),\bv[2]), \dots, (\rho_{r}(k),\bv[r]))=1/m \cdot R^{\bfi}(\bv)$ for every $k \in [m]$. After running through all $R(\bv) \in \mathcal{C}_{\bfi}$ we set the remaining weights to 0.
The weights of those tuples that correspond to $R(\bv)$ were selected so that their contribution to $\val(\bfy_2,\bfa,h)$ is equal to the inner sum in the $\sa^1$ objective function (\ref{eq:objBLP}); therefore, the total value of $h$ is equal to $\opt^{\sa^1}\ia$. It follows that $\opt(\bfy_2,\bfa)\leq \opt^{\sa^1}\ia$. Moreover, the iterated degree  of a pair $(k,v)$ in $\bfy_2$ is the same as in $\bfy_1$ (note here that the second item above guarantees that repeated entries are handled correctly). It follows that $\bfy_1 \eqone \bfy_2$, and the proof is concluded.
\end{proof}

The dual fractional homomorphism $\bfi \to_{df} \bfy_1$, the equivalence $\bfy_1 \eqone \bfy_2$, and assignments $Y_2 \to A$ that witness that $\opt(\bfy_2,\bfa) \leq \opt^{\sa^1}\ia$ from the proof of Theorem~\ref{thm:Decomp} can all be naturally associated with rational matrices (of dimensions $I \times Y_1$, $Y_1 \times Y_2$, and $Y_2 \times A$, respectively). It can be calculated that the product of these matrices is a matrix associated to a solution to the $\saia$ linear program. This is why we regard Theorem~\ref{thm:Decomp} as a decomposition theorem.

\section{Main result}
We are ready to prove the main result.   
The appropriate generalization of invariance  under $\eqone$ (item~\ref{item:2thmold} in Theorem~\ref{thm:old}) is that if $\bfi \eqone \bfj$ and $\thresh \in \mathbb{Q}$, then it cannot happen that $(\bfi,\thresh)$ is a $\yes$-instance while $(\bfj,\thresh)$ is a $\no$-instance. Item \ref{item:2thmain} in the following theorem is a reformulation of this requirement. 

\begin{theorem} \label{th:main}
Let $\ab$ be a promise valued template of signature $\sigma$. Then the following are equivalent.
\begin{enumerate}[label=\textnormal{(\roman*)}]
\item There exists a distributed algorithm that solves $\pvcsp\ab$. 
    Moreover, in such a case, there is a polynomial-time distributed algorithm that solves $\pvcsp\ab$.
    \label{item:1thmain}
    \item For all finite-valued $\sigma$-structures $\bfi, \bfj$, if $\bfi \eqone \bfj$ then $\opt(\bfj, \bfb) \leq \opt(\bfi, \bfa)$. \label{item:2thmain}
    \item $\sa^1$ decides $\pvcsp\ab$. \label{item:3thmain}
\end{enumerate}
\end{theorem}

\begin{proof}
\ref{item:1thmain} $\Rightarrow$ \ref{item:2thmain}. From the nature of the distributed model, it follows that agents with the same iterated degree will be in the same state at any time during the execution of any distributed algorithm. Therefore, if \ref{item:1thmain} holds and $\bfi \eqone \bfj$ are connected, then a terminating distributed algorithm will report the same decision when run on input $(\bfi,\thresh)$ or $(\bfj,\thresh)$ (see Proposition 2.2 and Corollary 2.3 in \cite{ButtiD2021distributed}), so by setting $\thresh=\opt\ia$ we obtain that \ref{item:2thmain} holds for all connected $\bfi$, $\bfj$. We now show how \ref{item:2thmain} in its full generality follows from \ref{item:2thmain} restricted to connected $\bfi$ and $\bfj$. 

Let us call two finite-valued $\sigma$-structures $\bfi$ and $\bfj$  \emph{weakly congruent} if $|J| \cdot \delta(\bfi) = |I| \cdot \delta(\bfj)$.
We claim that $\opt(\bfj,\bfb)/|J| \leq \opt(\bfi,\bfa)/|I|$ whenever $\bfi$ and $\bfj$ are weakly congruent and connected. The claim clearly holds when $|I|=1$ or $|J|=1$, so assume $|I|,|J| \geq 2$.
For any positive integer $k$, we define connected finite-valued $\sigma$-structures $\bfi'^{(k)}$ and $\bfi^{(k)}$ (and similarly $\bfj'^{(k)}$, $\bfj^{(k)}$) as follows. Let $I = \{v_0, v_1, \dots, v_{|I|-1}\}$ and  let the universe of $\bfi'^{(k)}$ be $\{0,1, \dots, k-1\} \times I$. Let $\eta$ be the probability distribution over the mappings $I \to I^{(k)}$ assigning probability $1/2k$ to each of the $2k$ mappings $f_j$, $f'_j$, $j \in \{0, 1, \dots, k-1\}$, where $f_j(v_i) = (j,v_i)$ and $f_j'(v_i) = (i+j \mod k, v_i)$ for each $v_i \in I$. We define the weights in $\bfi'^{(k)}$ in the unique way so that (\ref{eq:dualFH}) holds for $\eta$ with equality instead of inequality. 
Then $\eta$ is a dual fractional homomorphism from $\bfi$ to $\bfi'^{(k)}$, and the probability distribution which assigns probability 1 to the projection onto $I$ is a dual fractional homomorphism in the opposite direction. By Proposition~\ref{prop:dualFH}, $\opt(\bfi,\bfc)=\opt(\bfi'^{(k)},\bfc)$ for any valued $\sigma$-structure $\bfc$. Finally, let $\bfi^{(k)}$ be the valued $\sigma$-structure obtained from $\bfi'^{(k)}$ by multiplying weights by $2k$; clearly,  $\opt(\bfi^{(k)},\bfc) = 2k \opt(\bfi'^{(k)},\bfc)$ for any $\bfc$.
It follows from the construction that $\bfi^{(k)}$ is connected. Moreover, if $k$ is large enough ($k \geq |I|$ suffices), then the iterated degree of $(j,v_i)$ in $\bfi^{(k)}$ is obtained from the iterated degree of $v_i$ in $\bfi$ by multiplying all the variable multisets in each of the elements of $\delta^\bfi(v_i)$ by 2 (in each inductive step in the definition of iterated degree).
It follows that, for all $k'$, the valued structures $\bfj^{(k'|I|)}$ and $\bfi^{(k'|J|)}$ are connected and, when $k'$ is large enough, have the same iterated degree.  By item~\ref{item:2thmain} for connected valued structures, we get   
$\opt(\bfj^{(k'|I|)},\bfb) \leq \opt(\bfi^{(k'|J|)},\bfa)$ and the claim follows using the equalities above and rearranging.

Before finishing the proof, notice a simple consequence of the definition of iterated degrees. For a ``variable vertex'' $x$ of $\graph{\bfi}$, a label $S$, and a ``constraint vertex'' $y$ such that $x$ and $y$ are adjacent in $\graph{\bfi}$, denote $x[S,y] = \{y' \mid \delta(y')=\delta(y), \lbl_{\{x,y'\}} = S\}$. Observe that if there exists an edge between $x$ and $y$ labeled $S$, then $\{x'[S,y] \mid \delta(x') = \delta(x)\}$ is a collection of mutually disjoint sets of equal size, which cover $\{y' \mid \delta(y)=\delta(y')\}$; and, moreover, the same claim holds  when $x'$ and $y'$ are restricted to the connected component containing $x$ (or $y$). It follows that for a component $\bfi'$ of $\bfi$ and a component $\bfj'$ of $\bfj$, where $\bfi \eqone \bfj$, either the iterated degrees $\delta(\bfi')$ and $\delta(\bfj')$ are disjoint, or $\bfi'$ and $\bfj'$ are weakly congruent.

This observations allows us to finish the proof as follows. Let $\bfi$ and $\bfj$ be finite-valued $\sigma$-structures such that $\bfi \eqone \bfj$ and let $n=|I|=|J|$. Then there are sequences  $(\bfi_1, \dots, \bfi_n)$ and $(\bfj_1, \dots, \bfj_n)$ such that the first (resp., second) sequence contains each connected component $\bfi'$ of $\bfi$ (resp., $\bfj'$ of $\bfj$) exactly $|I'|$ times (resp., $|J'|$ times), and
$\bfi_i$ and $\bfj_i$ are weakly congruent for every $i \in [n]$. From the claim above, we get $\opt(\bfj_i,\bfb)/|J_i| \leq \opt(\bfi_i,\bfa) / |I_i|$ for every $i \in [n]$. Summing up these inequalities and observing that $\opt(\bfi,\bfa)$ is equal to the sum of $\opt(\bfi',\bfa)$ over all connected components $\bfi'$ of $\bfi$ (and similarly for $\bfj$), item~\ref{item:2thmain} now follows.

 \ref{item:2thmain} $\Rightarrow$ \ref{item:3thmain}. 
 We need to show that for every non-negative finite-valued $\sigma$-structure $\bfi$, $\opt \ib \leq \opt^{\sa^1}\ia$.
 Let $\bfy_1, \bfy_2$ be the structures obtained from Theorem~\ref{thm:Decomp}, i.e., $\yone \eqone \ytwo$,  $\opt(\bfy_2,\bfa) \leq \opt^{\sa^1}\ia$, and there is a dual fractional homomorphism from $\bfi$ to $\bfy_1$. Then, by \ref{item:2thmain} we have that $\opt(\bfy_1,\bfb) \leq \opt^{\sa^1}\ia$, and by Proposition \ref{prop:dualFH}, $\opt(\bfi,\bfb) \leq \opt^{\sa^1}\ia$ too, as required.

\ref{item:3thmain} $\Rightarrow$ \ref{item:1thmain}.
From  Theorem 3.2 in \cite{ButtiD2021distributed} (adapted to the valued setting), if $\opt^{\sa^1}\ia < \infty$, then there is a solution to the linear program that assigns the same value to every class of variables and constraints of $\bfi$ that have the same iterated degree.%
\footnote{We remark here that this theorem also has a substantially simpler proof -- it is enough to observe that averaging over variables and constraints with the same iterated degrees does not increase the objective function.}  
This allows us to reduce the linear program as follows.
Let $I\quot$ and $\mathcal{C}_{\bfi}\quot$ denote  the sets of equivalence classes of variables and constraints, respectively, under the equivalence $\eqone$. 
The new linear program, denoted $\sa^1_\equiv\ia$, contains one variable $p_{[v]}(a)$ for every class $[v] \in I\quot$ and one variable $p_{[R(\textbf{v})]}(\textbf{a})$ for every class $[R(\textbf{v})] \in \mathcal{C}_\bfi\quot$. The variables of the new program $\sa^1_\equiv\ia$ are subject to the same constraints as in $\sa^1\ia$, except they use the new reduced set of variables. The new objective function is 
\begin{equation}
\opt^{\sa^1_{\equiv}}\ia := \min \sum_{[R(\bv)] \in \mathcal{C}_\bfi\quot} k_{[R(\bv)]} \sum_{\textbf{a} \in A^{\ar(R)}}  p_{[R(\textbf{v})]}(\textbf{a})  R^\bfi(\textbf{v}) R^\bfa(\textbf{a}), \label{eq:objSA1equiv} \end{equation}
where $k_{[R(\bv)]} =|[R(\bv)]|$ is the number of constraints equivalent to $R(\bv)$. By the above discussion, we have $\opt^{\sa^1_{\equiv}}\ia = \opt^{\sa^1}\ia$.
Therefore, since $\sa^1$ decides $\pvcsp\ab$, so does $\sa^1_\equiv$.
(We remark here that two input structures with the same iterated degree have the same reduced $\sa^1_{\equiv}$ up to renaming of variables; this can be used e.g. to show  that \ref{item:3thmain} implies \ref{item:2thmain}.)

In order to show that \ref{item:3thmain} implies \ref{item:1thmain}, assume that $\bfi$ is a connected input structure. We show that  every agent in the distributed network can obtain the reduced linear program of $\sa^1_{\equiv}$ via a polynomial-time distributed algorithm. As $\sa^1_\equiv$ decides $\pvcsp\ab$, this will conclude the proof.

The agents can calculate their iterated degree (or, rather, a finite and effectively computable representation thereof) in polynomial time using a simple distributed version of the color refinement algorithm. Each agent $\alpha(x), x \in I \cup \mathcal{C}_\bfi$ can then use the representation of the iterated degree as an identifier, see~\cite[Lemma 4.8.]{ButtiD2021distributed} for a more detailed discussion.
 Every agent can obtain sufficient information from its neighbours to compute the equations in (\ref{eq:SA1}), (\ref{eq:SA2}), (\ref{eq:SA3}) and (\ref{eq:SA4}) that constrain its relevant LP variables of the reduced system (and use the identifiers to name the LP variables), and can subsequently broadcast these along the network. We are left to show that every agent can also compute the objective function of $\sa^1_{\equiv}\ia$. In fact, it is sufficient that every agent $\alpha(R(\bv))$ computes the summand of $\opt^{\sa^1_{\equiv}}\ia$ that corresponds to $[R(\bv)]$ and then broadcasts it in order to obtain the complete objective function. The only nontrivial piece of information to compute is the value of the coefficients $k_{[R(\bv)]}$.

By the observation made in the proof of \ref{item:1thmain} $\Rightarrow$ \ref{item:2thmain}, for each $R(\bv) \in \mathcal{C}_{\bfi}$, a participating variable $v \in \{\bv\}$, and $S = \lbl_{\{v,R(\bv)\}}$, the coefficient $k_{[R(\bv)]}$ is equal to the number of $S$-labeled edges from  $v$ into members of $[R(\bv)]$  (denoted $v[S,R(\bv)]$ above) multiplied by the size of $[v]$.
The former value can be computed by $\alpha(v)$, so $\alpha(v)$ can compute the ratio $k_{[R(\bv_1)]}:k_{[R(\bv_2)]}$ for any two constraints $R(\bv_1), R(\bv_2)$ that $v$ participates in. After broadcasting all these ratios, each agent can compute the ratios between any two $k_{[R(\bv)]}$ and, since the sum of these coefficients is $|\mathcal{C}_{\bfi}|$ (which is known to the agents),
they can compute the coefficients.
\end{proof}
Clearly, the implication \ref{item:4thmain} $\Rightarrow$ \ref{item:3thmain} in Theorem~\ref{thm:old} remains true for PVCSP (so the equivalent statements in Theorem~\ref{th:main} are satisfied in, e.g., the PVCSPs in Example~\ref{ex:pvcsp}). The following example shows that, unlike for PCSPs, the converse implication does not hold in general: we provide an example of a $\pvcsp$ template that is decided by $\sa^1$ but not by $\blp$.

\begin{example} \label{ex:SAnotBLP}
Let $\bfa$, $\bfb$ be $\sigma$-structures where $\sigma$ contains a single binary relation symbol $R$. Let $A=B=\{0,1\}$,  $R^\bfa(a,a)=R^\bfb(a,a)=3$ for $a \in \{0,1\}$, and $R^\bfa(a,b)=2$, $R^\bfb(a,b)=0$ for $a \neq b \in \{0,1\}$. The probability distribution which assigns probability 1 to the identity function is a fractional homomorphism, and so $\ab$ is a $\pvcsp$ template.

We claim that $\blp$ does not decide $\pvcsp\ab$. Indeed, let $\bfi$ be the $\pvcsp$ input structure given by $I=\{v\}$ and $R^{\bfi}(v,v)=1$. Then, there is a feasible solution to $\blp\ia$ given by $p_v(a)=1/2$ for $a \in \{0,1\}$ and $p_{R(v,v)}(a,a)=0$, $p_{R(v,v)}(a,b)=1/2$ for $a \neq b \in \{0,1\}$. This solution witnesses that $\opt^\blp\ia \leq 2$, however, it is easy to see that $\opt\ib =3$ and so $\blp$ does not decide $\pvcsp\ab$.

On the other hand, we show that  $\opt\ib \leq \opt^{\sa^1}\ia$ for any input valued structure $\bfi$. Let $V_l(\bfi) = \sum_{v \in I} R^\bfi(v,v)$ and $V_e(\bfi) = \sum_{u \neq v} R^\bfi(u,v)$ be the total weight of the constraints in $\bfi$ with and without repetitions, respectively. We choose an assignment $h:I \to B$ at random: each $h(v)$ is chosen independently and uniformly (both 0 and 1 with probability $1/2$).
The expected value of $\val(\bfi,\bfb,h)$ is $3 V_l(\bfi) + 3/2V_e(\bfi) $, which implies that $\opt\ib\leq 3 V_l(\bfi) + 3/2V_e(\bfi)$. 
As for $\sa^1$, we know that any feasible solution must have $p_{(v,v)}(a,b)=0$ whenever $a \neq b$. Therefore, we get
\begin{align*}
    \opt^{\sa^1}\ia & = \min \Big[ \sum_{v \in I}  \sum_{a \in A}  p_{R(v,v)}(a,a) R^\bfi(v,v) R^\bfa(a,a)   + \\ & \  \sum_{u \neq v \in I} \sum_{a,b \in A} p_{R(u,v)}(a,b) R^\bfi(u,v) R^\bfa(a,b) \Big]
      \geq 3 V_l(\bfi)  + 2 V_e(\bfi) > \opt\ib.
\end{align*}
\end{example}

\section{Conclusion}

We have shown that solvability of a PVCSP by the $\sa^1$ relaxation is equivalent to invariance under the Weisfeiler-Leman-like equivalence $\eqone$, and also to solvability in a natural distributed model.
The distributed algorithm for the narrower CSP setting from~\cite{ButtiD2021distributed} worked also for the search version of the problem, but this is unfortunately not the case for the algorithm presented in this paper. Is there an algorithm solving the search version of $\pvcsp\ab$ whenever the PVCSP is solvable by $\sa^1$? Note that in the search version an instance consists only of $\bfi$ and the goal is to find an assignment $h: I \to B$ such that $\val(\bfi,\bfb,h) \leq \opt \ia$. 

Another open problem emerges from Example~\ref{ex:SAnotBLP} which shows that $\blp$ and $\sa^1$ are not equivalent for PVCSPs. It follows from~\cite{ButtiD21fractional} that $\blp$ and $\sa^1$ are equivalent for PCSPs and from~\cite{ThapperZ16} that they are also equivalent for finite-valued VCSPs. Are these relaxations equivalent for general-valued VCSPs?

\paragraph{Acknowledgements} The authors are grateful to Victor Dalmau for his valuable comments.

\bibliography{biblio}

\begin{thebibliography}{10}

\bibitem{Barto2021algebraicPCSP}
Libor Barto, Jakub Bul{\'{\i}}n, Andrei~A. Krokhin, and Jakub Opr{\v s}al.
\newblock Algebraic approach to promise constraint satisfaction.
\newblock {\em J. {ACM}}, 68(4):28:1--28:66, 2021.
\newblock \href {https://doi.org/10.1145/3457606} {\path{doi:10.1145/3457606}}.

\bibitem{Barto2017polymorphisms}
Libor Barto, Andrei Krokhin, and Ross Willard.
\newblock {Polymorphisms, and How to Use Them}.
\newblock In Andrei Krokhin and Stanislav {\v Z}ivn{\' y}, editors, {\em The
  Constraint Satisfaction Problem: Complexity and Approximability}, volume~7 of
  {\em Dagstuhl Follow-Ups}, pages 1--44. Schloss Dagstuhl--Leibniz-Zentrum
  fuer Informatik, Dagstuhl, Germany, 2017.
\newblock URL: \url{http://drops.dagstuhl.de/opus/volltexte/2017/6959}, \href
  {https://doi.org/10.4230/DFU.Vol7.15301.1}
  {\path{doi:10.4230/DFU.Vol7.15301.1}}.

\bibitem{Bulatov2017}
A.~A. {Bulatov}.
\newblock A dichotomy theorem for nonuniform {CSPs}.
\newblock In {\em 2017 IEEE 58th Annual Symposium on Foundations of Computer
  Science (FOCS)}, pages 319--330, Oct 2017.
\newblock \href {https://doi.org/10.1109/FOCS.2017.37}
  {\path{doi:10.1109/FOCS.2017.37}}.

\bibitem{ButtiD2021distributed}
Silvia Butti and Victor Dalmau.
\newblock The complexity of the distributed constraint satisfaction problem.
\newblock In Markus Bl{\"{a}}ser and Benjamin Monmege, editors, {\em 38th
  International Symposium on Theoretical Aspects of Computer Science, {STACS}
  2021, March 16-19, 2021, Saarbr{\"{u}}cken, Germany (Virtual Conference)},
  volume 187 of {\em LIPIcs}, pages 20:1--20:18. Schloss Dagstuhl -
  Leibniz-Zentrum f{\"{u}}r Informatik, 2021.
\newblock \href {https://doi.org/10.4230/LIPIcs.STACS.2021.20}
  {\path{doi:10.4230/LIPIcs.STACS.2021.20}}.

\bibitem{ButtiD21fractional}
Silvia Butti and V{\'{\i}}ctor Dalmau.
\newblock {Fractional Homomorphism, Weisfeiler-Leman Invariance, and the
  Sherali-Adams Hierarchy for the Constraint Satisfaction Problem}.
\newblock In Filippo Bonchi and Simon~J. Puglisi, editors, {\em 46th
  International Symposium on Mathematical Foundations of Computer Science,
  {MFCS} 2021, August 23-27, 2021, Tallinn, Estonia}, volume 202 of {\em
  LIPIcs}, pages 27:1--27:19. Schloss Dagstuhl - Leibniz-Zentrum f{\"{u}}r
  Informatik, 2021.
\newblock \href {https://doi.org/10.4230/LIPIcs.MFCS.2021.27}
  {\path{doi:10.4230/LIPIcs.MFCS.2021.27}}.

\bibitem{Carbonnel2022otherside}
Clément Carbonnel, Miguel Romero, and Stanislav Živný.
\newblock The complexity of general-valued constraint satisfaction problems
  seen from the other side.
\newblock {\em SIAM Journal on Computing}, 51(1):19--69, 2022.
\newblock \href {http://arxiv.org/abs/https://doi.org/10.1137/19M1250121}
  {\path{arXiv:https://doi.org/10.1137/19M1250121}}, \href
  {https://doi.org/10.1137/19M1250121} {\path{doi:10.1137/19M1250121}}.

\bibitem{feder1998computational}
Tom{\'a}s Feder and Moshe~Y Vardi.
\newblock The computational structure of monotone monadic {SNP} and constraint
  satisfaction: A study through datalog and group theory.
\newblock {\em SIAM Journal on Computing}, 28(1):57--104, 1998.

\bibitem{Fioretto2018}
Ferdinando Fioretto, Enrico Pontelli, and William Yeoh.
\newblock Distributed constraint optimization problems and applications: A
  survey.
\newblock {\em J. Artif. Int. Res.}, 61(1):623--698, January 2018.

\bibitem{Kazda21}
Alexander Kazda.
\newblock Minion homomorphisms give reductions between promise valued {CSPs}.
\newblock 2021. In preparation.

\bibitem{Kolmogorov2017}
Vladimir Kolmogorov, Andrei~A. Krokhin, and Michal Rol{\'{\i}}nek.
\newblock The complexity of general-valued {CSPs}.
\newblock {\em {SIAM} J. Comput.}, 46(3):1087--1110, 2017.
\newblock \href {https://doi.org/10.1137/16M1091836}
  {\path{doi:10.1137/16M1091836}}.

\bibitem{KolmogorovTZ15}
Vladimir Kolmogorov, Johan Thapper, and Stanislav {\v Z}ivn{\'{y}}.
\newblock The power of linear programming for general-valued {CSPs}.
\newblock {\em {SIAM} J. Comput.}, 44(1):1--36, 2015.
\newblock \href {https://doi.org/10.1137/130945648}
  {\path{doi:10.1137/130945648}}.

\bibitem{KozikConsistency}
Marcin Kozik.
\newblock {Solving CSPs Using Weak Local Consistency}.
\newblock {\em SIAM Journal on Computing}, 50(4):1263--1286, 2021.
\newblock \href {http://arxiv.org/abs/https://doi.org/10.1137/18M117577X}
  {\path{arXiv:https://doi.org/10.1137/18M117577X}}, \href
  {https://doi.org/10.1137/18M117577X} {\path{doi:10.1137/18M117577X}}.

\bibitem{krokhin2017constraint}
Andrei Krokhin and Stanislav {\v Z}ivn{\'y}.
\newblock {\em The Constraint Satisfaction Problem: Complexity and
  Approximability}, volume~7.
\newblock Schloss Dagstuhl, 2017.

\bibitem{KZ_vcsps}
Andrei Krokhin and Stanislav {\v Z}ivn{\' y}.
\newblock {The Complexity of Valued CSPs}.
\newblock In Andrei Krokhin and Stanislav {\v Z}ivn{\' y}, editors, {\em The
  Constraint Satisfaction Problem: Complexity and Approximability}, volume~7 of
  {\em Dagstuhl Follow-Ups}, pages 233--266. Schloss Dagstuhl--Leibniz-Zentrum
  fuer Informatik, Dagstuhl, Germany, 2017.
\newblock URL: \url{http://drops.dagstuhl.de/opus/volltexte/2017/6966}, \href
  {https://doi.org/10.4230/DFU.Vol7.15301.233}
  {\path{doi:10.4230/DFU.Vol7.15301.233}}.

\bibitem{Kun13}
G{\'{a}}bor Kun.
\newblock Constraints, {MMSNP} and expander relational structures.
\newblock {\em Comb.}, 33(3):335--347, 2013.
\newblock \href {https://doi.org/10.1007/s00493-013-2405-4}
  {\path{doi:10.1007/s00493-013-2405-4}}.

\bibitem{Kun2012}
Gabor Kun, Ryan O’Donnell, Suguru Tamaki, Yuichi Yoshida, and Yuan Zhou.
\newblock Linear programming, width-1 {CSPs}, and robust satisfaction.
\newblock In {\em Proceedings of the 3rd Innovations in Theoretical Computer
  Science Conference}, ITCS ’12, page 484–495, New York, NY, USA, 2012.
  Association for Computing Machinery.
\newblock \href {https://doi.org/10.1145/2090236.2090274}
  {\path{doi:10.1145/2090236.2090274}}.

\bibitem{leman1968reduction}
AA~Leman and B~Weisfeiler.
\newblock A reduction of a graph to a canonical form and an algebra arising
  during this reduction.
\newblock {\em Nauchno-Technicheskaya Informatsiya}, 2(9):12--16, 1968.

\bibitem{Ramana1994}
Motakuri~V. Ramana, Edward~R. Scheinerman, and Daniel Ullman.
\newblock Fractional isomorphism of graphs.
\newblock {\em Discrete Mathematics}, 132(1):247 -- 265, 1994.
\newblock URL:
  \url{http://www.sciencedirect.com/science/article/pii/0012365X94902410},
  \href {https://doi.org/https://doi.org/10.1016/0012-365X(94)90241-0}
  {\path{doi:https://doi.org/10.1016/0012-365X(94)90241-0}}.

\bibitem{Scheinerman2011fractional}
Edward~R Scheinerman and Daniel~H Ullman.
\newblock {\em Fractional graph theory: a rational approach to the theory of
  graphs}.
\newblock Courier Corporation, 2011.

\bibitem{farkas}
Alexander Schrijver.
\newblock {\em Theory of Linear and Integer Programming}.
\newblock John Wiley \& Sons, Inc., USA, 1986.

\bibitem{Sherali1990}
Hanif~D. Sherali and Warren~P. Adams.
\newblock A hierarchy of relaxations between the continuous and convex hull
  representations for zero-one programming problems.
\newblock {\em {SIAM} J. Discret. Math.}, 3(3):411--430, 1990.
\newblock \href {https://doi.org/10.1137/0403036} {\path{doi:10.1137/0403036}}.

\bibitem{ThapperZ12}
Johan Thapper and Stanislav {\v Z}ivn{\'{y}}.
\newblock The power of linear programming for valued {CSPs}.
\newblock In {\em 53rd Annual {IEEE} Symposium on Foundations of Computer
  Science, {FOCS} 2012, New Brunswick, NJ, USA, October 20-23, 2012}, pages
  669--678. {IEEE} Computer Society, 2012.
\newblock \href {https://doi.org/10.1109/FOCS.2012.25}
  {\path{doi:10.1109/FOCS.2012.25}}.

\bibitem{ThapperZ16}
Johan Thapper and Stanislav {\v Z}ivn{\'{y}}.
\newblock The complexity of finite-valued {CSPs}.
\newblock {\em J. {ACM}}, 63(4):37:1--37:33, 2016.
\newblock \href {https://doi.org/10.1145/2974019} {\path{doi:10.1145/2974019}}.

\bibitem{ThapperZ17}
Johan Thapper and Stanislav {\v Z}ivn{\'{y}}.
\newblock {The Power of Sherali-Adams Relaxations for General-Valued CSPs}.
\newblock {\em {SIAM} J. Comput.}, 46(4):1241--1279, 2017.
\newblock \href {https://doi.org/10.1137/16M1079245}
  {\path{doi:10.1137/16M1079245}}.

\bibitem{ViolaZ21}
Caterina Viola and Stanislav {\v Z}ivn{\'{y}}.
\newblock The combined basic {LP} and affine {IP} relaxation for promise
  {VCSPs} on infinite domains.
\newblock {\em {ACM} Trans. Algorithms}, 17(3):21:1--21:23, 2021.
\newblock \href {https://doi.org/10.1145/3458041} {\path{doi:10.1145/3458041}}.

\bibitem{yokoo1992distributed}
Makoto Yokoo, Toru Ishida, Edmund~H Durfee, and Kazuhiro Kuwabara.
\newblock Distributed constraint satisfaction for formalizing distributed
  problem solving.
\newblock In {\em [1992] Proceedings of the 12th International Conference on
  Distributed Computing Systems}, pages 614--621. IEEE, 1992.

\bibitem{Zhu20}
Dmitriy Zhuk.
\newblock A proof of the {CSP} dichotomy conjecture.
\newblock {\em J. {ACM}}, 67(5):30:1--30:78, August 2020.
\newblock \href {https://doi.org/10.1145/3402029} {\path{doi:10.1145/3402029}}.

\end{thebibliography}

\end{document}